\documentclass[copyright,creativecommons]{eptcs}

\usepackage{iftex}

\ifpdf
  \usepackage{underscore}         
  \usepackage[T1]{fontenc}        
\else
  \usepackage{breakurl}           
\fi

\ifdefined \VersionLong
	\newcommand{\LongVersion}[1]{#1}
	\newcommand{\ShortVersion}[1]{}
\else
	\newcommand{\LongVersion}[1]{}
	\newcommand{\ShortVersion}[1]{#1}
\fi

\usepackage{orcidlink}
\newcommand{\orcidID}[1]{\orcidlink{#1}}
\usepackage{fontawesome}
\newcommand{\homepage}[1]{\href{#1}{\color{gray}\faHome}}

\usepackage[utf8]{inputenc}
\usepackage[english]{babel} 
\usepackage{csquotes}
\usepackage{lmodern} 
\usepackage[T1]{fontenc} 

\usepackage{amsthm}
\usepackage{thm-restate}
\usepackage{xspace}
\usepackage{adjustbox}
\usepackage{amsmath}
\usepackage{amssymb} 
\usepackage{booktabs}

\usepackage{subcaption} 

\usepackage[misc,geometry]{ifsym} 

\usepackage{pifont} 

\ifdefined\VersionWithComments
	\usepackage[colorinlistoftodos,textsize=footnotesize]{todonotes}
\else
	\usepackage[disable]{todonotes}
\fi
\newcommand{\gennote}[3]{\todo[size=\scriptsize,linecolor=#2,backgroundcolor=#2!25,bordercolor=#2]{#3: #1}}  

\newcommand{\sd}[1]{{\gennote{#1}{green}{SD}}}

\usepackage{tikz}
\usetikzlibrary{arrows,automata,positioning,shapes} 
\usetikzlibrary{decorations} 
\usetikzlibrary{decorations.pathreplacing} 
\usetikzlibrary{calc} 

\tikzstyle{pta}=[auto, ->, >=stealth', /tikz/initial distance=10pt]
\tikzstyle{every node}=[initial text=]
\tikzstyle{RA-location}=[rectangle, rounded corners, minimum size=12pt, draw=black, fill=blue!10, inner sep=2pt]
\tikzstyle{location}=[rounded rectangle, minimum size=12pt, draw=black, fill=blue!10, inner sep=2pt]
\tikzstyle{invariant}=[draw=black, dotted, inner sep=1pt, node distance=0] 
\tikzstyle{final}=[double, fill=blueColorBlind!40]

\tikzstyle{private}=[fill=redColorBlind!50,thick]

\definecolor{coloract}{rgb}{0, 0.3, 0}
\definecolor{colorcact}{rgb}{0.6, 0.5, 0}
\definecolor{colorclock}{rgb}{0.4, 0, 0}
\definecolor{colorloc}{rgb}{0.1, 0.1, 0.1}

\definecolor{blueColorBlind} {RGB}{68 ,119,170}
\definecolor{greenColorBlind}{RGB}{34 ,136,51 }
\definecolor{redColorBlind}  {RGB}{238,102,119}

\ifdefined\VersionWithComments
\usepackage{everypage}

\AddEverypageHook{%
	\begin{tikzpicture}[remember picture, overlay]
		\node[
		align=center,
		anchor=north east,
		xshift=-1cm,
		yshift=-1cm,
		text=red
		] at (current page.north east) {\scalebox{2}{\bfseries\LARGE Draft!}\\{\tiny\today{}}};
	\end{tikzpicture}
}

\fi

\usepackage{graphicx}
\definecolor{LSdeuxNrichblack}{HTML}{020A13}
\definecolor{LSdeuxNbazul}{HTML}{1273CE}
\definecolor{LSdeuxNcornflower}{HTML}{5297FF}
\definecolor{LSdeuxNAntiflashwhite}{HTML}{F8FBFF}

\usepackage{paralist} 

\newenvironment{ienumerate}
	{\ifdefined\VersionLong\begin{enumerate}\else\begin{inparaenum}[\itshape i\upshape)]\fi}
	{\ifdefined\VersionLong\end{enumerate}\else\end{inparaenum}\fi}

\usepackage[pagewise]{lineno} 

\theoremstyle{plain}

\newtheorem{theorem}{Theorem}

\theoremstyle{definition}
\newtheorem{definition}{Definition}
\newtheorem{example}{Example}

\theoremstyle{remark}

\usepackage[capitalise,english,nameinlink]{cleveref} 
\crefname{line}{\text{line}}{\text{lines}} 

\newcommand{\defProblem}[3]
{%
	\noindent\fcolorbox{LSdeuxNrichblack}{LSdeuxNAntiflashwhite}{
	\begin{minipage}{.95\columnwidth}
		\textbf{#1 problem:}\\
		\textsc{Input}: #2\\
		\textsc{Problem}: #3
	\end{minipage}
}

	\smallskip

}

\newif\ifinfigure
\let\infigure\iffalse 

\let\oldtikzpicture\tikzpicture
\let\endoldtikzpicture\endtikzpicture
\renewenvironment{tikzpicture}
  {\let\ifinfigure\iftrue\oldtikzpicture}
  {\endoldtikzpicture\let\ifinfigure\iffalse}

\newcommand{\styleact}[1]{\ensuremath{
	\ifinfigure%
		\textcolor{coloract}{{#1}}%
	\else%
		#1%
	\fi%
	}%
}

\newcommand{\stylecact}[1]{\ensuremath{
	\ifinfigure%
		\textcolor{colorcact}{{#1}}%
	\else%
		#1%
	\fi%
	}%
}

\newcommand{\styleclock}[1]{\ensuremath{
	\ifinfigure%
		\textcolor{colorclock}{{#1}}%
	\else%
		#1%
	\fi%
	}%
}

\newcommand{\styleloc}[1]{\ensuremath{
	\ifinfigure%
		\textcolor{colorloc}{{#1}}%
	\else%
		#1%
	\fi%
	}%
}

\newcommand{\assign}{\leftarrow}

\newcommand{\checkUseMacro}[1]{#1}

\newcommand{\set}[1]{\ensuremath{\left\{#1\right\}}}

\newcommand{\setN}{\ensuremath{\mathbb{N}}}

\newcommand{\setR}{\ensuremath{\mathbb{R}}}
\newcommand{\setRgeqzero}{\ensuremath{\setR_{\geq 0}}}

\newcommand{\compOp}{\bowtie}

\newcommand{\init}{\ensuremath{0}}
\newcommand{\priv}{\ensuremath{{\mathit{priv}}}}

\newcommand{\final}{\ensuremath{f}}

\newcommand{\styleAutomaton}[1]{\ensuremath{\mathcal{#1}}}

\newcommand{\clock}{\ensuremath{\styleclock{x}}}

\newcommand{\clocki}[1]{\ensuremath{\styleclock{\clock_{#1}}}}
\newcommand{\ClockCard}{H} 
\newcommand{\clockval}{\ensuremath{\mu}}
\newcommand{\ClockSet}{\ensuremath{\mathbb{X}}} 
\newcommand{\ClocksZero}{\ensuremath{\vec{0}}}

\newcommand{\resets}{\ensuremath{R}}
\newcommand{\reset}[2]{\ensuremath{[#1]_{#2}}}

\newcommand{\TA}{\ensuremath{\checkUseMacro{\styleAutomaton{A}}}}
\newcommand{\TB}{\ensuremath{\checkUseMacro{\styleAutomaton{B}}}} 

\newcommand{\actionu}{\ensuremath{\styleact{u}}}
\newcommand{\caction}{\ensuremath{\stylecact{\kappa}}} 
\newcommand{\cactioni}[1]{\ensuremath{\stylecact{\caction_{#1}}}}
\newcommand{\ActionSet}{\ensuremath{\Sigma}}

\newcommand{\extraAction}{\ensuremath{\styleact{\sharp}}}

\newcommand{\constraint}{\ensuremath{C}}

\newcommand{\edge}{\ensuremath{\checkUseMacro{e}}}
\newcommand{\edgei}[1]{\ensuremath{\checkUseMacro{\edge_{#1}}}}
\newcommand{\EdgeSet}{\ensuremath{E}}

\newcommand{\guard}{\ensuremath{g}}

\newcommand{\invariant}{\ensuremath{I}} 

\newcommand{\loc}{\ensuremath{\styleloc{\ell}}}
\newcommand{\loci}[1]{\ensuremath{\styleloc{\loc_{#1}}}}
\newcommand{\locinit}{\ensuremath{\styleloc{\loc_\init}}}
\newcommand{\locfinal}{\ensuremath{\styleloc{\loc_\final}}}
\newcommand{\FinalSet}{\ensuremath{L_f}}
\newcommand{\PrivSet}{\ensuremath{L_{\mathit{priv}}}}
\newcommand{\locpriv}{\ensuremath{\styleloc{\loc_\priv}}}

\newcommand{\LocSet}{\ensuremath{L}}

\newcommand{\longuefleche}[1]{\stackrel{#1}{\longrightarrow}}

\newcommand{\run}{\checkUseMacro{\rho}} 

\newcommand{\duration}{\ensuremath{\mathit{dur}}} 
\newcommand{\runduration}[1]{\ensuremath{\checkUseMacro{\duration}(#1)}}

\newcommand{\Dpriv}{\ensuremath{\mathcal{D}_\mathit{priv}}}
\newcommand{\Dpub}{\ensuremath{\mathcal{D}_\mathit{pub}}}

\newcommand{\strategy}{\ensuremath{\sigma}}

\newcommand{\TAcontrolextend}{\ensuremath{\left(\ControlSet, \LocSet, \locinit, \PrivSet, \FinalSet, \ClockSet, \invariant, \EdgeSet\right)}}

\newcommand{\ControlSet}{\ensuremath{\ActionSet_c}}

\newcommand{\semantics}[1]{\ensuremath{\mathfrak{T}_{#1}}}

\newcommand{\semanticscontrolextend}{\ensuremath{\left(\StateSet, \concstateinit, \ControlSet\cup\{\uncontrol\} \cup \setRgeqzero, \transition\right)}}

\newcommand{\transition}{{\ensuremath{\rightarrow}}}

\newcommand{\StateSet}{\ensuremath{\mathfrak{S}}}
\newcommand{\concstate}{\ensuremath{\mathfrak{s}}}

\newcommand{\concstateinit}{\ensuremath{\concstate_\init}}

\newcommand{\PrivVisit}{\ensuremath{\mathit{Visit}_{\mathit{priv}}}}
\newcommand{\PubVisit}{\ensuremath{\mathit{Visit}_{\mathit{pub}}}}

\newcommand{\Parts}[1]{\ensuremath{\mathcal{P}(#1)}}
\newcommand{\AllRuns}[2]{\ensuremath{\Omega}^{#2}(#1)}

\newcommand{\APriv}{\ensuremath{\TA_{\mathit{priv}}}}
\newcommand{\APub}{\ensuremath{\TA_{\mathit{pub}}}}

\newcommand{\interval}[2]{\ensuremath{[\![#1 ; #2]\!]}}

\newcommand{\TimeRegions}{\ensuremath{\mathit{TI}}}

\newcommand{\substrat}{\ensuremath{\nu}}

\newcommand{\op}{\ensuremath{(+, \cdot, ^2, \star)}}
\newcommand{\regexp}{\ensuremath{\mathit{exp}}}

\newcommand{\uncontrol}{\ensuremath{\stylecact{u}}}
\newcommand{\lab}[1]{\ensuremath{\mathit{label}(#1)}}

\newcommand{\ComplexityFont}[1]{{\sffamily\upshape #1}}

\newcommand{\NEXPTIME}{\ComplexityFont{NEXPTIME}\xspace}
\newcommand{\coNEXPTIME}{\ComplexityFont{Co-NEXPTIME}\xspace}

\newcommand{\NLOGSPACE}{\ComplexityFont{NLOGSPACE}\xspace}
\newcommand{\EXPSPACE}{\ComplexityFont{EXPSPACE}\xspace}

\newcommand{\fiveEXPTIME}{\ComplexityFont{5EXPTIME}\xspace}
\newcommand{\NP}{\ComplexityFont{NP}\xspace}
\newcommand{\coNP}{\ComplexityFont{Co-NP}\xspace}

\newcommand{\ie}{i.e.,\xspace}
\newcommand{\st}{s.t.}

\title{
Execution-Time Opacity in (Controllable) Timed Automata is Hard\thanks{
	This work, as well as the long version~\cite{DL26}, is partially supported by ANR BisoUS (ANR-22-CE48-0012) and ANR GUMMIS 
	(ANR-25-CE48-5096) .
}
}
\author{Sarah Dépernet\orcidID{0009-0003-8710-7934}
\qquad\qquad Engel Lefaucheux\homepage{https://elefauch.github.io/}%
	\orcidID{0000-0003-0875-300X}
\institute{
Université de Lorraine, CNRS, Inria, LORIA, F-54000 Nancy, France
}
}
\def\titlerunning{ET-Opacity in (Controllable) TA is Hard}
\def\authorrunning{S.\ Dépernet and E.\ Lefaucheux}

\begin{document}

\sloppy

\pagestyle{plain}

\maketitle

\begin{abstract}
Timing information leakage occurs whenever an attacker successfully deduces confidential 
internal information by observing some timed information such as the total duration of a 
run.
Timed automata (TAs)
are an extension of finite-state automata with a set of clocks evolving at the same rate and that can be tested or reset. Hence, this formalism is able to represent systems involving timing constraints reacting to the passage of time.\sd{Remarque du reviewer 2 appliquée} 
Some of the transitions in a timed automaton may be controlled, i.e., enabled or disabled, to limit the runs that are allowed.

Execution-time opacity (ET-opacity) aims at studying timing information leakage through the analysis of timed automata.
More precisely, the ET-opacity problem consists in deciding whether an 
attacker can deduce whether some secret locations were visited by a run, solely from the run's duration.
When control is allowed, the goal becomes to find a strategy that blocks the runs leaking information.
This problem has been shown decidable for TAs, with and without control.
However, these existing results only provide upper bounds on the complexity of 
the ET-opacity problem, sometimes largely above the tight bound.
Our main contribution is to provide tight lower bounds for those two results.
Our proofs rely on novel constructions.

Moreover, for controllable TAs, two variants of ET-opacity (depending on whether detecting if a run did \emph{not} visit a secret location represents an information leakage or not, and named full and weak ET-opacity respectively) had been separately studied; we prove here that these two notions are inter-reducible.
\end{abstract}

\section{Introduction}

Consider a cryptographic server that receives a request, performs some secret computation, and replies. An adversary who cannot read the internal state of the server, cannot see which branch of the code was taken, which key was used, which data was accessed, may nonetheless learn something simply by measuring how long the server takes to respond. This is the essence of a \emph{timing side-channel attack}: confidential information leaks not through the content of an output, but through its timing.
Timed automata (TAs)~\cite{AD94} offer a natural formalism for modelling such
real-time systems. They extend finite automata with a finite set of real-valued clocks, all evolving at the same rate, which can be tested against integer thresholds and reset. This allows one to faithfully represent systems subject to timing constraints, and to reason about the time elapsed during a computation.
Within this setting, \emph{opacity} formalises the requirement that an external
attacker, armed with a given observation power, cannot infer whether a designated
secret was involved in a run.

Depending on what the attacker can observe, a range of opacity notions were
studied with considerable variations in their decidability status.
Cassez' seminal work~\cite{Cassez09} showed that opacity is undecidable in dense time
when the attacker can observe the full timed trace produced by a run, even for
restricted subclasses of TAs such as event-recording automata~\cite{AFH99}.
This motivates restricting either the model or the observation power.
Along the first direction, decidability has been recovered for real-time
automata~\cite{WZ18,WZA18}, for constant-time labelled automata~\cite{LLHL22}, and for
labelled real-timed automata~\cite{Zhang24},
but also for some TAs with one-clock~\cite{ADL24} or to a discrete-time semantics~\cite{AGWZH24,ADL24,KKG24}.
Time-bounded variants, where the attacker must disclose the secret before a
deadline~\cite{AEYM21}, and memory-bounded attackers~\cite{SLR23} have also been
considered.

Along the second direction, opacity was considered when the attacker has a limited
accuracy~\cite{DQY26}, or access to only part of the observation~\cite{ALMS22,ADL24}.
We follow this latter direction, assuming the
attacker has access solely to the \emph{total execution time} of a run (the aggregate
duration from start to termination) and nothing else.
This \emph{execution-time opacity} (ET-opacity) framework was introduced
in~\cite{ALMS22} and later developed in~\cite{ALM23,ALLMS23}.

In this setting, a TA is weakly ET-opaque if every duration achievable by a
private run (one visiting a secret location) is also achievable by some public run:
the attacker cannot conclude from the duration alone that a secret was visited.
Full ET-opacity additionally requires the symmetric inclusion: every public duration is
also a private one, so the attacker cannot conclude that a secret was \emph{not}
visited either.
In stark contrast to Cassez' setting, most ET-opacity problems for TAs are decidable~\cite{ALLMS23}.

In~\cite{ABLM22,ADLL26}, the ET-opacity framework was extended with \emph{control}: some transitions of the TA are governed by a \emph{sequential strategy} that dynamically selects which controllable actions are available at each moment. 
The resulting controller is deliberately based on time alone: rather than observing the current location or clock valuation and reacting to the internal state of the system, it specifies the availability of controllable actions as a function of global time. This can be viewed as an external scheduler that regulates when certain operations may occur, without access to the internal execution whose timing may reveal the secret. In other words, the controller does not require the system to leak information to it, which may be intercepted by an attacker.

Intuitively, a sequential strategy divides time into consecutive intervals of a fixed granularity, represented here by unit intervals, and within each interval prescribes an ordered agenda of sets of permitted actions. The strategy therefore specifies the succession of control decisions, while their exact timing is left to an implementation. In this sense, sequential strategies provide a finite and structured description of time-based control, without requiring the strategy itself to specify arbitrary real-valued switching times. This restricted form is also essential from an algorithmic perspective:~\cite{ADLL26} showed that if strategies were allowed to be fully time-accurate, \emph{i.e.}, to enable actions at arbitrary real-valued timestamps, the control problem becomes undecidable. The sequential structure thus identifies a restricted, yet expressive, class of schedulers for which the synthesis problem remains tractable.

In the controllable framework, the question shifts from checking whether a given TA is ET-opaque to deciding whether \emph{there exists} a strategy that makes it so. This is highly relevant in practice: it models whether a real-time system can be \emph{scheduled} so as to prevent timing side-channel leakage. For instance, one may ask whether a server can be instrumented, by controlling the availability or ordering of operations, so that no execution time reveals whether a secret branch was taken.

For uncontrolled TAs, ET-opacity was first shown decidable in \fiveEXPTIME~\cite{ALMS22}, before a more careful analysis brought the upper bound down to \NEXPTIME{}~\cite{ALM23} (this result failed to take a negation into account however, their algorithm is in fact in \coNEXPTIME).
For controlled TAs, the existence of a sequential strategy ensuring ET-opacity was shown to lie in \EXPSPACE{}~\cite{ADLL26}, via a reduction to a one-player game on a doubly exponential region automaton. In both cases, however, \emph{no matching lower bounds} were established, leaving open whether these complexities are tight.
This article closes both gaps. 
Our contributions are threefold.

\begin{enumerate}
  \item \textbf{Inter-reducibility of weak and full ET-opacity under control.}
In the uncontrolled setting, weak and full ET-opacity are inter-reducible~\cite{ADL24}: one can construct from a TA $\TA$ another TA $\TA'$ with private and public runs swapped, so that $\TA$ is fully opaque if and only if both $\TA$ and $\TA'$ are weakly opaque. 

This argument breaks down under control: nothing guarantees that the two TAs can be made weakly opaque by a single strategy. We show, perhaps surprisingly, that inter-reducibility nonetheless holds in the controlled setting (\cref{prop:weak-full}), via a more subtle construction that internalises both inclusion checks within a single automaton. As a consequence, the complexity of weak and full ET-opacity coincides under control, and hardness results transfer freely between them.

  \item \textbf{\coNEXPTIME-completeness of ET-opacity (without control).}
    We prove that both weak and full ET-opacity for TAs are \coNEXPTIME-complete
    (\cref{ET-opac-unc}), matching (the correction of) the upper bound of~\cite{ALM23}.
    The hardness reduction builds on a known \NEXPTIME-hard fragment of parametric
    TCTL (Timed Computation Tree Logic) model-checking~\cite{QSW17} and encodes the negation of the relevant synchronisation property
    into ET-opacity.\sd{Remarque du reviewer 2 appliquée}

  \item \textbf{\EXPSPACE-completeness of ET-opacity control.}
    We prove that the problem of deciding whether a sequential strategy making a TA
    weakly (or fully) ET-opaque exists is \EXPSPACE-complete (\cref{th:ET-opac-comp}),
    matching the upper bound of~\cite{ADLL26}.
    The hardness reduction proceeds from the \EXPSPACE-complete inequivalence problem
    for regular expressions with squaring~\cite{SM73}, relying on the ability to produce 
    a doubly exponential countdown gadget.
\end{enumerate}
Full proofs are available in the long version~\cite{DL26}.

\section{Preliminaries}\label{sec:preliminaries} 

We denote by $\setN$ and $\setRgeqzero$ the sets of non-negative integers and non-negative reals, respectively.
If $a$ and~$b$ are two integers with $a \leq b$, the set $\set{a, a+1, \dots, b-1, b}$ is denoted by $\interval{a}{b}$.
\emph{Clocks} are real-valued variables that all evolve over time at the same rate.
Throughout this paper, we assume a set~$\ClockSet = \{ \clocki{1}, \dots, \clocki{\ClockCard} \} $ of clocks.\sd{Remarque du reviewer 3 appliquée}
A \emph{clock valuation} is a function
$\clockval : \ClockSet \rightarrow \setRgeqzero$, assigning a non-negative value to each clock.
We write $\ClocksZero$ for the clock valuation assigning $0$ to all clocks.
Given a constant $d \in \setRgeqzero$, $\clockval + d$ denotes the valuation \st\ $(\clockval + d)(\clock) = \clockval(\clock) + d$, for all $\clock \in \ClockSet$. If $R$ is a subset of $\ClockSet$ and $\clockval$ a clock valuation, we call \emph{reset} of the clocks of $\resets$ and denote by $[\clockval]_R$ the valuation \st\ for all clock $x \in \ClockSet$, $\reset{\clockval}{\resets}(x) = 0$ if $x \in \resets$ and $\reset{\clockval}{\resets}(x) = \clockval(x)$ otherwise.
A \emph{constraint}~$\constraint$ is a conjunction of inequalities over~$\ClockSet $ of the form
$\clock \compOp d$, with
${\compOp} \in \{<, \leq, =, \geq, >\}$
and
$d \in \setN$.
Given~$\constraint$, we write~$\clockval \models \constraint$ if the expression obtained by replacing each~$\clock$ with~$\clockval(\clock)$ in~$\constraint$ evaluates to true.
We denote $\top$ the constraint that is always true.

\subsection{Timed automata}\label{sec:TA}

Timed automata (TAs) extend finite automata with a finite set of clocks with values in~$\setRgeqzero$.
We extend standard TAs with
\begin{ienumerate}%
	\item a set of private locations (needed to model the secret in our subsequent opacity notions), and
	\item a set of controllable events. 
\end{ienumerate}%

\LongVersion{\begin{definition}[Timed automaton (TA)]}
\ShortVersion{\begin{definition}}
\label{def:TA}
	A TA~$\TA$ is a tuple \mbox{$\TA = \TAcontrolextend$}, where:
	\begin{enumerate}%
		\item $\ControlSet$ is a finite set of controllable actions;
		\item $\LocSet$ is a finite set of locations,
		\LongVersion{\item }$\locinit \in \LocSet$ is the initial location,
		\item $\PrivSet \subseteq \LocSet$ is a set of private locations,
		\item $\FinalSet \subseteq \LocSet$ is a set of final locations,
		\item $\ClockSet$ is a finite set of clocks,
		\item $\invariant$ is the invariant, assigning to every $\loc\in \LocSet$ a constraint $\invariant(\loc)$ and such that $\ClocksZero~\models~\invariant(\locinit)$,
		\item $\EdgeSet$ is a finite set of edges  $\edge = (\loc,\guard,\caction,\resets,\loc')$
		where~$\loc,\loc'\in \LocSet$ are the source and target locations, 
		$\caction \in \ControlSet \cup \{ \uncontrol \}$ (where $\uncontrol$ denotes an uncontrollable action),
		$\resets\subseteq \ClockSet$ is a set of clocks to be reset, and $\guard$ is a constraint 
			(called \emph{guard}).
	\end{enumerate}%
\end{definition}

\begin{figure}[b]
\begin{minipage}{\textwidth/2}
	\centering
	\normalsize

	\begin{tikzpicture}[pta, scale=2, xscale=1, yscale=.5]

		\node[location, initial] at (0, -1) (s0) {$\loci{0}$};

		\node[location, private] at (1, 0) (s2) {$\loci{2}$};

		\node[location, final] at (2, -1) (s1) {$\loci{1}$};

		\node[invariant, below=of s2] {$\styleclock{\clock} < 1$};

		\path (s0) edge[bend left] node[below right, align=center]{$\cactioni{1}$} (s2);
		\path (s0) edge[loop above] node[above, align=center]{$\styleclock{\clock} = 2$\\$\styleclock{\clock} \assign 0$\\$\cactioni{2}$} (s0);
		\path (s0) edge[] node[below, align=center]{$\uncontrol$\\$\styleclock{\clock} \leq 1$} (s1);
		\path (s2) edge[bend left] node[above, align=center]{$\styleclock{\clock}\leq 1$\\$\cactioni{2}$} (s1);

	\end{tikzpicture}
	\captionof{figure}{A TA example}
	\label{figure:example-TA}
\end{minipage}
\begin{minipage}{\textwidth/2}
	\begin{tikzpicture}[x=4.5cm, y=1cm]
	\def\t{0.45}; 
	\def\a{\t / 2}; 
	\def\b{.5 + \a}; 
	\def\colorA{blue!80!black};
	\def\colorB{green!50!black};
	\def\tcolorA{\textcolor{\colorA}};
	\def\tcolorB{\textcolor{\colorB}};

	\draw[->] (-0.1,0) -- (1.1,0) node[right] {$t$};

	\draw[very thick] (0, 0.2) -- (0, -0.2);
	\node[below] at (0, -0.2) {$n$};
	\draw[very thick] (1, 0.2) -- (1, -0.2);
	\node[below] at (1, -0.2) {$n+1$};
	
	\large	
	\node at (0.01,0) {\textbf{\tcolorA{(}}};
	\node at (\t,0) {\textbf{\tcolorA{)}}};
	\draw[thick, color=\colorA] (0,0) -- (\t, 0);
	
	\node at (\t + 0.01,0) {\textbf{\tcolorB{[}}};
	\node at (0.99,0) {\textbf{\tcolorB{)}}};
	\draw[thick, color=\colorB] (\t,0) -- (1, 0);

	
	\normalsize
	\node[above] at (0,1) {$\ActionSet_{n,0} = \emptyset$};
	\node[above] at (\a,0.1) {\tcolorA{$\ActionSet_{n,1} = \set{\cactioni{1}}$}};
	\node[above] at (\b,0.1) {\tcolorB{$\ActionSet_{n,2} = \set{\cactioni{2}}$}};
	\node[above] at (1,1) {$\ActionSet_{n+1,0} = \emptyset$};
	\node[below] at (\t,-0.15) {$t_1 = n + $\t};
	
	\draw[dashed] (0,0) -- (0,1);
	\draw[dashed] (1,0) -- (1,1);
	\end{tikzpicture}
	\captionof{figure}{An implementation of the \\strategy $\strategy_1$ from $\cref{exa:strategy}$}
	\label{figure:implementation}
\end{minipage}
\end{figure}

\begin{example}
	In \cref{figure:example-TA}, we give an example of a TA with three locations $\loci{0}$, $\loci{1}$ and~$\loci{2}$,
	four edges,
	two controllable actions~$\{ \cactioni{1}, \cactioni{2} \}$,
	and one clock $\clock$.
	The location $\loci{0}$ is initial, $\loci{2}$ is the (unique) private location, and~$\loci{1}$ is the (unique) final location.
The location $\loci{2}$ has an invariant ``$\styleclock{\clock} < 1$''. The edge from $\loci{0}$ to $\loci{1}$ is labelled by the uncontrollable action~$\uncontrol$, and has a guard ``$\styleclock{\clock} \leq 1$''. The self-loop of $\loci{0}$, labelled by the controllable action~$\cactioni{2}$, resets the clock $\styleclock{\clock}$ and has a guard $\styleclock{\clock} = 2$.
\end{example}

\begin{definition}[Semantics of a TA]\label{def:semantics}
	Given a TA $\TA = \TAcontrolextend$,
	the semantics of~$\TA$ is given by the timed transition system $\semantics{\TA} = \semanticscontrolextend$, with
	\begin{enumerate}
		\item $\StateSet = \big\{ (\loc, \clockval) \in \LocSet \times \setRgeqzero^\ClockSet \mid \clockval \models \invariant(\loc) \big\}$,
		$\concstateinit = (\locinit, \ClocksZero) $,
		\item  $\transition \subseteq \StateSet \times \EdgeSet \times \StateSet \cup \StateSet \times \setRgeqzero \times \StateSet$ consists of the discrete and (continuous) delay transition relations:
		\begin{enumerate}
			\item discrete transitions:
$\big((\loc,\clockval), \edge, (\loc',\clockval') \big) \in \transition$,
		if $(\loc, \clockval) , (\loc',\clockval') \in \StateSet$, $\edge = (\loc,\guard,\caction,\resets,\loc') \in \EdgeSet$, $\clockval'= \reset{\clockval}{\resets}$, and $\clockval\models \guard$.
			\item delay transitions:
$\big((\loc,\clockval), d, (\loc,\clockval +d)\big) \in \transition$,
		if $d \in \setRgeqzero$ and
$\forall d' \in [0, d], (\loc, \clockval+d') \in \StateSet$.
		\end{enumerate}
	\end{enumerate}
\end{definition}

Moreover we write $(\loc, \clockval)\longuefleche{(d, \edge)} (\loc',\clockval')$ for a combination of a delay and discrete transition if
$\exists  \clockval'' :  \big( (\loc,\clockval), d, (\loc,\clockval'') \big) \in \transition$
and
$\big( (\loc,\clockval'') , \edge,  (\loc',\clockval') \big) \in \transition$. 

Given a TA~$\TA$ with semantics $\semanticscontrolextend$, we refer to the elements of~$\StateSet$ as the \emph{configurations} of~$\TA$.
A (finite) \emph{run} of~$\TA$ is an alternating sequence of configurations of~$\TA$ and pairs of delays and edges starting from the initial configuration $\concstateinit$ and ending in a final configuration (\ie{} whose location is final),
of the form
$(\loci{0}, \clockval_{0}), (d_0, \edge_0), (\loci{1}, \clockval_{1}), \ldots (\loci{n}, \clockval_{n})$ for some $n \in \setN$,
with
$\loci{n} \in \FinalSet$ and for $i = 0, 1, \dots n-1$, $\loci{i} \notin \FinalSet$, $\edge_i \in \EdgeSet$, $d_i \in \setRgeqzero, $ and
$(\loci{i}, \clockval_{i}) \longuefleche{(d_i, \edge_i)} (\loci{i+1}, \clockval_{i+1})$.
We denote by $\AllRuns{\TA}{}$ the set of runs of~$\TA$.
A \emph{path} of~$\TA$ is an infix of a run beginning with a configuration and ending with a configuration.\sd{Remarque du reviewer 3 appliquée}

Given a TA~$\TA$ and a run~$\run=(\loci{0}, \clockval_0), (d_0, \edgei{0}), (\loci{1}, \clockval_1), \ldots, (\loci{n}, \clockval_n)$ of~$\TA$ with $(\locinit, \clockval_0)$ the initial configuration, we say that $\run$ \emph{visits} the location $\loc \in \LocSet$ if there exists~$m \in \interval{0}{n}$ such that $\loci{m} = \loc$. If $\run$ visits a private location, \ie{} a location in $\PrivSet$, we say it is a \emph{private run}, and otherwise it is a \emph{public run}.
We denote by~$\PrivVisit(\TA)$ the set of private runs, and by~$\PubVisit(\TA)$ the set of public runs.

In our setting and concerning a run $\run$, an attacker may have access to the total execution time $\duration(\run)$, which is the sum of all delays appearing in the run.
This notation is extended to sets of runs.\sd{Remarque du reviewer 3 appliquée}

\subsection{Control}\label{sec:def-strategies}

The ET-opacity framework can be extended with the notion of control: some transitions of the TA, those labelled by controllable actions in $\ControlSet$, can be enabled or disabled at runtime in order to enforce opacity. This is modelled by a \emph{sequential strategy} (first introduced in \cite{ADLL26} under the name ``meta-strategy''), which prescribes, for each time interval, an ordered list of sets of controllable actions that are made available one after another within that time interval\sd{Remarque du reviewer 3 appliquée}. Crucially, the strategy commits to this agenda based solely on the current time interval, without observing the internal state of the system.

Formally, we call \emph{time intervals} the intervals of the form $\set{n}$ or $(n;n+1)$ where $n \in \setN$, and denote their set $\TimeRegions$. 
A sequential strategy is a function mapping time intervals to a finite sequence from 
$\Parts{\ControlSet}^*$,
such that if the time interval is a singleton, then the sequence returned by the strategy contains only one element:

\[
\begin{array}{rrcl}
	\strategy{}: & \TimeRegions  & \longrightarrow & \mathcal{P}(\ControlSet)^* \\
	
	& \set{n} & \longmapsto & \ActionSet_{n,0} \\
	& (n; n+1) & \longmapsto & \ActionSet_{n,1} \ActionSet_{n,2} \dots \ActionSet_{n,K_n}
 
\end{array}
\]

Limiting the sequence to a single element over a singleton represents the inability for a 
system to change its behaviour in zero time. 
The actual behaviour of a sequential strategy is described by an implementation of it:
an \emph{implementation of a sequential strategy $\strategy$} refines $\strategy$ with exact timestamps.
That is, for each time interval $I =(n;n+1)$,
given the sequence $\strategy(I)=\ActionSet_{1} \ActionSet_{2} \dots \ActionSet_{K_n}$ with $K_n \geq 1$,
an implementation $\substrat$ produces a timed sequence $(\ActionSet_{1}, t_{1}) (\ActionSet_{2}, t_{2}) \dots (\ActionSet_{K_n}, t_{K_n})$ with $(t_{ k})_{k \in \interval{1}{K_n}}$ an increasing sequence with values in~$I$, $t_{K_n}=n+1$. Intuitively, the implementation $\substrat$ enables the controllable
actions in $\ActionSet_{i}$ within the interval $[t_{i-1};t_{i})$ if $i > 1$, and $(n; t_1)$ if $i = 1$.
Formally, given~$\substrat$ some implementation of a strategy~$\strategy$, $t \in \setRgeqzero$ and~$I$ time interval containing $t$, let $\substrat(I)=(\ActionSet_{1}, t_{1}) (\ActionSet_{2}, t_{2}) \dots (\ActionSet_{K_n}, t_{K_n})$,
the set of control labels \emph{enabled by~$\substrat$ at time~$t$} is
\[ E_{\substrat}(t) = \ActionSet_{i} \cup \set{\uncontrol} \; \text{ with } i = \min
 \set{k \mid t \leq t_{k} }.\]

The runs of $\TA$ whose transitions are all enabled by an implementation~$\substrat$ of the strategy~$\strategy$ are called \emph{$\strategy$-compatible} runs.
Their set is denoted $\AllRuns{\TA}{\strategy}$. We use the notation $\lab{e}$ to denote the control label of the edge $e$.
Formally:
\begin{align*}
\AllRuns{\TA}{\strategy} &= \lbrace \run \in \AllRuns{\TA}{} \mid \run = (\locinit, \clockval_{0})(d_0,\edgei{0})\dots (d_{n-1},\edgei{n-1})(\loci{n}, \clockval_{n}):
\exists \substrat \text{ implementing }\\
& \strategy \text{ such that } 
 \forall 0\leq i< n, \lab{\edgei{i}} \text{ is enabled }\text{by } \substrat \text{ at time } \sum_{k=0}^i d_k
\rbrace
\end{align*}


\begin{example}\label{exa:strategy}
Consider the TA from \cref{figure:example-TA} and let $\strategy_1$ be the strategy such that for all $n \in \setN$, $\strategy_1( (n; n+1) ) = \set{\cactioni{1}} \set{\cactioni{2}}$ and $\strategy_1( \set{n} ) = \emptyset$. 
Intuitively, within each unit interval, $\kappa_1$
is available first and $\kappa_2$ second, while no controllable action is available at integer times. 
As a consequence, the self-loop on $\ell_0$, which requires $\kappa_2$ at time $x=2$, an integer, is not enabled.
The only $\strategy_1$-compatible runs are those that take the uncontrollable transition at a time at most 1, and those that go through $\loci{2}$ at a time in $(0;1)$, let some positive time elapse, and then reach $\loci{1}$ strictly before the time~$1$. 
A possible implementation of $\strategy_1$ 
is represented on \cref{figure:implementation}.
\end{example}
%
%

\subsection{ET-Opacity of a controlled TA}\label{sec:winning-condition}

We first define ET-opacity for an uncontrolled TA, before extending to the controlled setting. 
We denote $\Dpriv(\TA) = \duration(\PrivVisit(\TA))$ the set of durations of private runs of~$\TA$, and $\Dpub(\TA) = \duration(\PubVisit(\TA))$ for the set of durations of public runs of~$\TA$.
ET-opacity formalises the requirement that whatever the attacker observes (here, the total execution time of a run) they cannot determine whether a private location was visited.\sd{Typo corrigée}
We say that a duration is \emph{ET-opaque} if it can be obtained from both a public and a private run, \ie{} whenever it is in $\Dpriv(\TA) \cap \Dpub(\TA)$.
We say a timed automaton $\TA$ is \emph{weakly ET-opaque} if $\Dpriv(\TA) \subseteq \Dpub(\TA)$ (every private duration\sd{Typo corrigée} is also public) and \emph{fully ET-opaque} if $\Dpriv(\TA) =  \Dpub(\TA)$ (the durations of every run of $\TA$ are ET-opaque). We also occasionally use the term ET-opaque for isolated durations.

\medskip 

\defProblem{Full (resp.\ weak) ET-opacity of a TA}{A TA $\TA$}{Is $\TA$ fully ET-opaque (resp.\ weakly ET-opaque)?}

\medskip

These definitions extends naturally to the controlled setting.
Given a sequential strategy $\strategy$, we denote by~$\PrivVisit^\strategy(\TA)$ the set of private runs in $\AllRuns{\TA}{\strategy}$, and by~$\PubVisit^\strategy(\TA)$ the set of public runs in $\AllRuns{\TA}{\strategy}$, and set 
$\Dpriv^\strategy(\TA) = \duration(\PrivVisit^\strategy(\TA))$ and $\Dpub^\strategy(\TA) = \duration(\PubVisit^\strategy(\TA))$.
A timed automaton $\TA$ is then \emph{weakly ET-opaque under $\strategy$} if $\Dpriv^{\strategy}(\TA) \subseteq \Dpub^\strategy(\TA)$ and \emph{fully ET-opaque under $\strategy$} if $\Dpriv^\strategy(\TA) =  \Dpub^\strategy(\TA)$.
In this case, we say that $\strategy$ makes $\TA$ weakly (resp. fully) ET-opaque, or that $\strategy$ is \emph{opacity-ensuring} for $\TA$.

\medskip 

\defProblem{Full (resp.\ weak) ET-opacity of a controlled TA}{A TA $\TA$}{Does there exist a sequential strategy making $\TA$ fully ET-opaque (resp.\ weakly ET-opaque)?}

\medskip

\begin{example}\label{exa:strategy-weak}
Consider the TA $\TA$ from \cref{figure:example-TA} under the sequential strategy
$\strategy_1$ defined in \cref{exa:strategy}.
Consider the duration $0.5$ for instance and the implementation $\nu$ of $\strategy_1$
(illustrated on \cref{figure:implementation}) such that for all $n \in \setN$, $\nu$ switches from $\set{\cactioni{1}}$ to $\set{\cactioni{2}}$ at the time $n + 0.45$ (\emph{i.e.} for all $n$, $\nu((n;n+1)) = (\set{\cactioni{1}}, n+0.45)(\set{\cactioni{2}},n+1)$ and $\nu(\set{n}) = \emptyset$).
The transitions of the private run which go from $\loc_0$ to $\loc_2$ at time $0.4$ and then to $\loc_1$ at time $0.5$ are enabled by $\nu$. Thus, this run is $\strategy_1$-compatible and therefore $0.5\in\Dpriv^{\strategy_1}(\TA)$.
More generally, the sets of private and public durations are respectively $\Dpriv^{\strategy_1}(\TA)~=~(0;1)$ and $\Dpub^{\strategy_1}(\TA)~=~[0;1]$ so $\TA$ is weakly opaque under $\strategy_1$, but not fully opaque. 
\end{example}

%
%

\section{Inter-reducibility of weak and full opacity in controllable TA}\label{sec:weak-full}

In~\cite{ADL24}, the weak ET-opacity and full ET-opacity problems were shown to be inter-reducible for TA without control. In this section, we establish the same result in the controllable setting through a completely new proof since the one from~\cite{ADL24} does not extends to the controllable setting.

\begin{restatable}{proposition}{interreduc}\label{prop:weak-full}
The weak ET-opacity and full ET-opacity problems for controlled TA are inter-reducible.
\end{restatable}

\begin{proof}[Sketch of proof.]
The reduction from weak to full opacity is a direct adaptation of the
construction of~\cite{ADL24}: it does not use control in any essential way,
so we focus here on the harder direction, from full to weak opacity.

\paragraph{Why the uncontrolled argument fails.}
For a TA~$\TA$, write $\APriv$ and $\APub$ for the two TAs
isolating, respectively, the private and public runs of~$\TA$
(so $\runduration{\AllRuns{\APriv}{\strategy}} = \Dpriv^\strategy(\TA)$
and likewise for $\mathrm{pub}$, for any strategy~$\strategy$).
In the uncontrolled setting, full opacity of $\TA$ is checked by testing weak
opacity of $\APriv$ against $\APub$ \emph{and} of
$\APub$ against $\APriv$, using two separate automata.
Under control this breaks: each of the two automata might be made weakly
opaque by its own strategy, without any single strategy working for both.
We must therefore force \emph{one} strategy to answer both questions at once.

\paragraph{Idea: run both checks in the same automaton, on disjoint times.}
We build a single TA $\TA_{\mathrm{full}\to\mathrm{weak}}$ containing one copy
of $\APriv$ and one copy of $\APub$, and let every run
choose, nondeterministically and invisibly, which copy to enter first. Three
ingredients make this choice carry the information we need.

\begin{itemize}
  \item \textbf{Time-slicing.} All constants in $\APriv$ and
  $\APub$ are doubled, and a fresh clock~$x$, reset every two time
  units, confines all their transitions to intervals $[2n+1,2n+2)$. This
  leaves every interval $[2n,2n+1)$ free, to be used only for a final
  "duration shift" step.
  \item \textbf{A routing clock~$y$.} From the initial location, a run picks
  a branch (call it "private-labelled" or "public-labelled") and then a copy
  to enter ($\APriv$ or $\APub$). The clock~$y$ is reset
  exactly when the label and the copy \emph{agree} (private label entering
  the private copy, or public label entering the public copy), and left
  un-reset when they \emph{disagree}. So $y$ silently records a single bit:
  "did the label match the copy?".
  \item \textbf{Reading off that bit as a delay.} At the end of the run, an
  extra time unit is forced to elapse if and only if $y$ was reset on entry.
  Because of the time-slicing, this shift can only happen in the free
  intervals $[2n,2n+1)$, so it never interferes with the doubled durations
  produced inside the copies.
\end{itemize}

This produces exactly four kinds of runs, summarised below ($f(D) =
\{2n+1+t \mid n+t \in D\}$ encodes the doubling, and $f(D)+1$ the doubling
together with the extra shift):

\begin{center}
\begin{tabular}{lccc}
\toprule
label & copy entered & $y$ reset? & duration set contributed \\
\midrule
private & $\APriv$ & yes & $f(\Dpriv^\strategy(\TA)) + 1$ \\
private & $\APub$  & no  & $f(\Dpub^\strategy(\TA))$ \\
public  & $\APriv$ & no  & $f(\Dpriv^\strategy(\TA))$ \\
public  & $\APub$  & yes & $f(\Dpub^\strategy(\TA)) + 1$ \\
\bottomrule
\end{tabular}
\end{center}

Hence, under (the appropriately rescaled version of) any strategy~$\strategy$,
\[
\Dpriv^\strategy(\TA_{\mathrm{full}\to\mathrm{weak}})
  = \bigl(f(\Dpriv^\strategy(\TA))+1\bigr) \cup f(\Dpub^\strategy(\TA))\]
  and \[
\Dpub^\strategy(\TA_{\mathrm{full}\to\mathrm{weak}})
  = f(\Dpriv^\strategy(\TA)) \cup \bigl(f(\Dpub^\strategy(\TA))+1\bigr).
\]

\paragraph{Why this gives the reduction.}
The sets $f(D)$ and $f(D)+1$ are always disjoint (they live in disjoint time
intervals), so the inclusion
$\Dpriv^\strategy(\TA_{\mathrm{full}\to\mathrm{weak}}) \subseteq
\Dpub^\strategy(\TA_{\mathrm{full}\to\mathrm{weak}})$
splits, term by term, into the two inclusions
$f(\Dpub^\strategy(\TA)) \subseteq f(\Dpriv^\strategy(\TA))$ and
$f(\Dpriv^\strategy(\TA))+1 \subseteq f(\Dpub^\strategy(\TA))+1$.
Since $f$ is injective, this is equivalent to
$\Dpriv^\strategy(\TA) = \Dpub^\strategy(\TA)$, \ie{} to full
ET-opacity of $\TA$ under~$\strategy$. So $\strategy$ (rescaled) makes
$\TA_{\mathrm{full}\to\mathrm{weak}}$ weakly ET-opaque if and only if it makes
$\TA$ fully ET-opaque, which is exactly the reduction we need.
\end{proof}

Note that this result implies the inter-reducibility result from \cite{ADL24} since the uncontrollable setting corresponds exactly to the controllable setting where all transitions are supposed uncontrollable. 
Hence, in the following proofs, we focus on weak ET-opacity and the hardness is carried onto full ET-opacity by the inter-reduction.

\section{Lower bounds for execution-time opacity problems}

ET-opacity for timed automata was introduced in~\cite{ALMS22} where it was shown to be decidable in \fiveEXPTIME (this result is more detailed in~\cite{ALLMS23}).
This algorithm relied on an expensive transcription of the problem into a decidable
logic.
It was later shown that most of the exponential blow ups could be avoided by translating the problem into the inclusion of languages between two (exponential) unary automata, 
producing a \coNEXPTIME algorithm but no lower bound to match~\cite{ALM23}.

Later, the ET-opacity framework was extended to controllable TA in~\cite{ADLL26} 
where the authors established that the existence of a sequential strategy ensuring full ET-opacity in a timed automaton was decidable. More precisely, the authors turn the 
TA into a doubly exponential one player game, built around a determinisation of the 
region automaton of the given TA. Solving the game could be done in \NLOGSPACE, hence producing an \EXPSPACE{} algorithm. Once again, no lower bound was established for this problem.

We refine here these two results by showing that their complexity are tight.

\begin{restatable}{theorem}{hardopac}\label{ET-opac-unc}
The weak (resp.full) ET-opacity of a TA problem is \coNEXPTIME-complete
\end{restatable}

\begin{proof}
We focus on weak ET-opacity and rely on~\cref{prop:weak-full} to extend the hardness to full ET-opacity.

Let us discuss the \coNEXPTIME algorithm first. 
In~\cite{ALM23}, the authors reduce ET-opacity to the language equality of two unary automata. The automata they build are of exponential size, and they claim the equality of langage can be tested (by~\cite{SM73}) in \NP, hence producing a \NEXPTIME algorithm. 
The claim is slightly wrong however:~\cite{SM73} considered non-equality, not equality, of langages.
Hence, by~\cite{SM73}, deciding the equality of the languages of two unary automata is in fact in \coNP. Therefore the algorithm of~\cite{ALM23} is in \coNEXPTIME.

We establish \coNEXPTIME-hardness now.
The reduction relies on a result of~\cite{QSW17}, who showed that model-checking a specific reachability formula of parametric TCTL is \NEXPTIME-hard\footnote{Remark: in~\cite{QSW17}, the authors do not rely on invariants, instead assuming that every infinite run has infinitely many transitions. This means that they dismiss every run reaching a configuration from which no transition will ever become available. 
In their hardness proof, their condition is used to force the system to take a transition
which reset a clock $x$ every 2 time units. 
Using an invariant which requires $x\leq 2$ hence has the same effect.
Thus while the frameworks are slightly different, this has no impact here.
}. 
The formula asks:
``given a TA $\TA$ and a set of target locations $P$, does there exist a time 
$\theta\geq 0$ such that every run of $\TA$ is in a location of $P$ at time exactly $\theta$?''.
Intuitively, this problem is hard because it requires witnessing a moment of global synchronisation across all runs: not just that some run reaches $P$ at time $\theta$, but that every run currently lie among those locations. 
The exponential blow-up stems from the fact that the synchronisation time can be doubly exponential in the size of the automaton.

We show how to encode this synchronisation property as the negation\footnote{Establishing that the negation of weak ET-opacity is \NEXPTIME-hard proving by definition that weak ET-opacity is \coNEXPTIME-hard.} of a weak ET-opacity question.
The key observation in this reduction is the following: if, from every location in $P$, a run can instantly terminate by first visiting a private location, and from every location outside $P$, a run can instantly terminate without visiting any private location, then the private durations are exactly the times at which a run is in $P$, and the public durations are exactly the times at which some run is outside $P$. 
A duration $\theta$ that is private but not public, a violation of weak ET-opacity, is therefore exactly a time at which every run is in $P$, \emph{i.e.}, a witness for the parametric TCTL formula. Conversely, if such a $\theta$ exists, it produces a private duration with no matching public run, violating weak ET-opacity.
The construction (illustrated in~\cref{fig:nexptime-reduction}) implements this idea. To ensure that the detour through the private location adds no time to the run's duration, the transition from $\loc_{\mathit{priv}}$ to $\loc_f$
is guarded by $z = 0$, where $z$ is a fresh clock reset upon entering $\loc_{\mathit{priv}}$. This forces the transition to be taken instantly, so the total duration of the private run is exactly $\theta$.

%
%
%

Formally, let $\TA=\TAcontrolextend$ be a TA and $P$ be a set of locations of $\TA$\sd{Typo corrigée}.
Note that as~\cite{QSW17} does not rely on control, private sets or final sets, we can assume that $\ControlSet = \PrivSet = \FinalSet=\emptyset$.
We build the TA $\TA'= \{\emptyset, \LocSet',\locinit,\{\locpriv\},\{\locfinal\},\ClockSet\cup\{z\},\invariant',\EdgeSet'\}$ (illustrated in \cref{fig:nexptime-reduction}) such that
\begin{itemize}
\item $\LocSet'=\LocSet\cup\{\locpriv,\locfinal\}$;
\item for all $\loc\in \LocSet, \invariant'(\loc)=\invariant(\loc)$ and 
$\invariant'(\locpriv)=\invariant'(\locfinal)=\top$;
\item 
$\EdgeSet'=\EdgeSet\cup\{(\loc,\top,\uncontrol, \emptyset,\locfinal)\mid \loc \in \LocSet\setminus P\}
 \cup \{(\loc,\top,\uncontrol, \{z\},\locpriv)\mid \loc \in P\} \cup \{(\locpriv,(z=0), \uncontrol, \emptyset, \locfinal)\}.$
\end{itemize}

\begin{figure}[bth]
\centering
\begin{tikzpicture}[pta, node distance=2cm, thin]

  \node[location, initial]  (l0)                       {$\locinit$};
  \node[location]           (lA) [right=1.7cm of l0]   {$\ell \in L \setminus P$};
  \node[location]           (lP) [above right=1.5cm and 2.2cm of l0] {$\ell \in P$};
  \node[location, private]  (lpriv) [right=5cm of lP]  {$\loc_{\mathit{priv}}$};
  \node[location, final]    (lf)  [below=1.5cm of lpriv] {$\loc_f$};

	\node[rectangle, minimum width=5cm, minimum height=2.9cm, draw] (rectangle) at (1.3cm, 0.9cm) {};
	\node[anchor=north west,inner sep=3pt] at (rectangle.north west) 
    {$\TA$};

  \path[->] (lA) edge
    node [below] {$\uncontrol$} (lf);

  \path[->] (lP) edge
    node [above] {$\uncontrol, \styleclock{z}\assign 0$} (lpriv);

  \path[->] (lpriv) edge
    node [right] {$\uncontrol, \styleclock{z} = 0$} (lf);

\end{tikzpicture}
\caption{TA $A'$ used in the \coNEXPTIME-hardness reduction. 
Locations in $P$ route through $\loc_{\mathit{priv}}$ (resetting
$z$), while other locations reach $\loc_f$ instantly.}	
\label{fig:nexptime-reduction}
\end{figure}
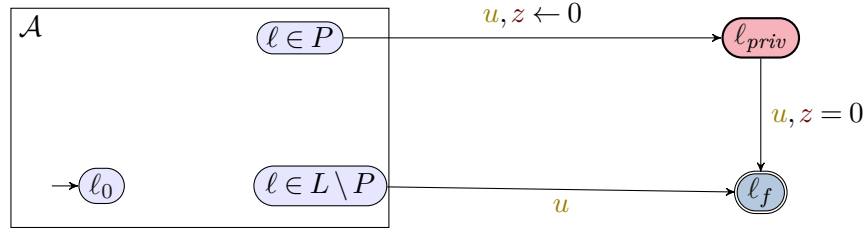

We now argue that $\TA'$ is not weakly ET-opaque if and only if there exists a
duration $\theta$ such that all runs of $\TA$ are in a location of $P$ at time
exactly $\theta$.

\smallskip
\noindent($\Rightarrow$)
Suppose $\TA'$ is not weakly ET-opaque, i.e., there exists a duration $d \in
\Dpriv(\TA') \setminus \Dpub(\TA')$.
Since $d \in \Dpriv(\TA')$, there is a private run of $\TA'$ of duration $d$.
A run is private iff it visits $\loc_{\mathit{priv}}$, which happens iff it exits
some location $\loc \in P$ via the transition resetting $z$.
The run then immediately (as the transition is guarded by a ``$z=0$'' guard) 
reaches $\loc_f$.
Hence $d = \theta$ for some $\theta$ at which the run was in a location of $P$.

Now since $d \notin \Dpub(\TA')$, no public run has duration $d$.
A public run of duration $d$ would exit some location $\loc \in \LocSet \setminus P$
at time $d$ via the instant transition to $\loc_f$, meaning some run of $\TA$ is in
$\LocSet \setminus P$ at time $d$.
But since no such public run exists, \emph{every} run of $\TA$ must be in a location
of $P$ at time $d$.

Hence all runs of $\TA$ are in a location of $P$ at time exactly $d$
(\emph{i.e.} $d$ is a witness for the formula of~\cite{QSW17}).

\smallskip
\noindent($\Leftarrow$)
Conversely, suppose there exists $\theta \geq 0$ such that every run of $\TA$ is
in a location of $P$ at time exactly $\theta$.
Then every run of $\TA'$ that reaches $\loc_f$ at time $\theta$ must have
passed through $\loc_{\mathit{priv}}$ (since the only way to reach $\loc_f$ at
time $\theta$ from a location of $P$ is to have taken the transition to $\locpriv$, resetting $z$, and immediately left for $\locfinal$).
All those runs are private, hence $\theta\in \Dpriv(\TA')$ and $\theta\not \in \Dpub(\TA')$, \emph{i.e.} $\TA'$ is not weakly ET-opaque.
\end{proof}

\begin{theorem}\label{th:ET-opac-comp}
The weak (resp.\ full) ET-opacity of a controlled TA problem is \EXPSPACE-complete.
\end{theorem}

\begin{proof}
We focus on weak ET-opacity, by \cref{prop:weak-full} the result then also applies to
full ET-opacity.
The \EXPSPACE{} upper bound was established in~\cite{ADLL26}. We prove
\EXPSPACE{}-hardness by reduction from the inequivalence problem for regular
expressions written with \op{} (union, concatenation, squaring, Kleene star), which
is \EXPSPACE{}-complete~\cite{SM73}.

\smallskip\noindent\textbf{Overview of the reduction.}

Given two regular expressions $\regexp_1$ and $\regexp_2$ over an alphabet $\ActionSet$,
we construct a TA~$\TA$ (represented in \cref{fig:difference-TA}) such that a
sequential strategy making $\TA$ weakly ET-opaque exists if and only if
$\regexp_1 \setminus \regexp_2 \neq \emptyset$. Since
$\regexp_1 \neq \regexp_2$ iff $\regexp_1 \setminus \regexp_2 \neq \emptyset$ or
$\regexp_2 \setminus \regexp_1 \neq \emptyset$, applying the construction twice
(with swapped indices) yields the full reduction.
The letters of $\ActionSet$, together with a new end symbol $\sharp$, will constitute the control alphabet $\ControlSet$ of the whole controllable TA.

By~\cite{ADL24}, every regular expression written with \op{} can be converted into a
discrete-time TA of polynomial size that produces words, one letter every two time units. 
We write $\TA_i= (\ControlSet, \LocSet_i,\locinit^i,\emptyset,\FinalSet^i,\ClockSet_i,\invariant_i,\EdgeSet_i)$ for the
TA obtained from $\regexp_i$, of polynomial size. In this construction, there is a clock, $x$, that is reset every two time units exactly when a transition labelled by an action in $\ActionSet$ is taken. We use this same clock in $\TA_1$, $\TA_2$ and refer to it as $x$ in the whole TA of the reduction.

For a word $w = w_0 w_1 \dots w_{|w|-1}$ in $\ActionSet^*$, we denote by $S_{w\sharp}$ the set of strategies enabling for $i$ in $\{0, \dots, |w|-1\}$ the single action $w_i$ at the time $2i$, the single action $\sharp$ at the time $2|w|$ and nothing during the rest of the time until the time $2|w|$.

The remainder of the construction will serve to constrain the form of the strategies ensuring weak ET-opacity on this TA, requiring that each opacity-ensuring strategy is in some $S_{w\sharp}$ with $w \in \regexp_1 \setminus \regexp_2$. In particular, several of the following blocks are made to prevent opacity-ensuring strategies from being of the following types:
\begin{itemize}
\item[] (type 1): the strategy enables a controllable action at a non even-integer time;
\item[] (type 2): the strategy enables two actions simultaneously;
\item[] (type 3): the strategy never enables $\sharp$.
\end{itemize}

The TA $\TA$ is obtained by building
a fresh initial location $\locinit$ from which four uncontrollable $\clock=0$
transitions lead to the initial locations of three blocks that run in parallel. Their roles are the following.
\begin{itemize}
  \item \emph{First block} $\TB_0$: a \emph{normalisation gadget} that penalises
    any strategy of (type 1) or (type 2). 
  \item \emph{Second block} ($\TB_1$ and $\TB_2$): two \emph{language-reading gadgets} built from $\TA_1$ and $\TA_2$, respectively. They read the word produced by the strategy and generate duration sets whose pattern reveals whether the word belongs to $\regexp_1 \setminus \regexp_2$.
  \item \emph{Third block} $\TB_{\extraAction}$: a \emph{countdown gadget} whose
    private-duration set is non-opaque after a doubly exponential time. This addresses the case of (type 3)-strategies by forcing any opacity-ensuring strategy to commit to a word $w\sharp$ before this deadline, making the reduction complete.
\end{itemize}
The complete TA $\TA$ is obtained by combining these four components as shown in
\cref{fig:difference-TA}. We now detail each block.

\begin{figure}[bt]
\begin{center}
\begin{tikzpicture}[pta, yscale=1.2, thin]

	\node[location, private] (lt1) at (-1, 4) {$\loc^{\textit{test}}$};
	\node[location] (lt2) at (+1, 4) {$\loc^{a}$};
	\node[location] (lt4) at (+5, 4) {$\loc^{\textit{end}}$};
	\node[location, final] (lt3) at (+8, 4) {};

	\node (l01) at (-1, 1.5) {};
	\node[location] (lf1) at (+1, 2.42) {$\locfinal^1$};
	\node[location] (lnf1) at (+1, 0.55) {};
	\node[rectangle, minimum width=3cm, minimum height=3cm, align=center, draw] (rectangle) at (0cm, 1.5cm) {$\TA_1$};
	
	\node[location] (l11) at (+5, 2.42) {$\loc^e$};
	\node[location, private] (l21) at (+5, 0.55) {$\locpriv^1$};
	\node[location, final] (l31) at (+8, 1.5) {};

	\node (l02) at (-1,-1.5) {};
	\node[location] (lf2) at (1, -1) {$\locfinal^2$};
	\node[location] (lnf2) at (+1, -2) {};
	\node[rectangle, minimum width=3cm, minimum height=2cm, align=center, draw] (rectangle) at (0cm, -1.5cm) {$\TA_2$};
	
	\node[location, private] (l12) at (+5, -1) {$\locpriv^2$};
	\node[location, private] (l22) at (+5, -2) {$\locpriv^3$};
	\node[location, final] (l32) at (+8, -1.5) {};
	
	\node (l03) at (-0.75,-3.5) {};
	\node[rectangle, minimum width=2.5cm, minimum height=1.5cm, align=center, draw] (rectangle) at (0cm, -3.5cm) {$\TA_{\extraAction}$};
	
	\node[location, initial] (l0) at (-4, 0) {$\locinit$};

	\path
	(l0) edge node[align=center, above left, pos=0.6]{$\actionu$\\$\styleclock{\clock} = 0$} (l01)
	(l0) edge[bend left] node[align=center, above left, pos=0.6]{$\actionu$\\$\styleclock{\clock} = 0$} (lt1)
	(l0) edge node[align=center, below=+3pt, pos=0.5]{$\actionu$\\$\styleclock{\clock} = 0$}  (l02)
	(l0) edge[bend right] node[below left, align=center, pos=0.6]{$\actionu$\\$\styleclock{\clock} = 0$} (l03)
    (lt1) edge[loop above] node[align=center]{$\actionu$\\$\styleclock{\clock} = 2$\\$\styleclock{\clock} \assign 0$} (lt1)
    (lt1) edge[bend left] node[align=center]{$\styleact{a}$\\$0<\styleclock{\clock} < 2$} (lt4)
    (lt1) edge node[align=center, below]{$\styleact{a}\in \ControlSet$\\$\styleclock{\clock} \assign 0$} (lt2)
    (lt2) edge node[align=center]{$\styleact{b}\neq a$\\$\styleclock{\clock} = 0$} (lt4)
    (lt4) edge node[align=center]{$\actionu$} (lt3)
    	(lf1) edge node[align=center]{$\styleact{\extraAction}$\\$\styleclock{\clock} = 2$\\$\styleclock{\clock} \assign 0$} (l11)
	(lf1) edge node[align=center, pos=0.7]{$\styleact{\extraAction}$\\$\styleclock{\clock} = 2$\\$\styleclock{\clock} \assign 0$} (l21)
	(lnf1) edge node[align=center, pos=0.3]{$\styleact{\extraAction}$\\$\styleclock{\clock} = 2$\\$\styleclock{\clock} \assign 0$} (l21)
	(l11) edge node[align=center]{$\actionu$\\$\styleclock{\clock} \geq 2$} (l31)
	(l21) edge node[align=center]{$\actionu$\\$\styleclock{\clock} \geq 2$} (l31)
	
	(lf2) edge node[align=center, above]{$\styleact{\extraAction}$\\$\styleclock{\clock} = 2$\\$\styleclock{\clock} \assign 0$} (l12)
	(lnf2) edge node[align=center, below]{$\styleact{\extraAction}$\\$\styleclock{\clock} = 2$\\$\styleclock{\clock} \assign 0$} (l22)
	(l12) edge node[align=center, above right]{$\actionu$\\$\styleclock{\clock} \geq 1$} (l32)
	(l22) edge node[align=center, below]{$\actionu$\\$\styleclock{\clock} \geq 2$} (l32)
	
	;
\end{tikzpicture}
\end{center}
  \caption{Representation of the reduction to the ET-opacity of a controlled TA}
  \label{fig:difference-TA}
\end{figure}
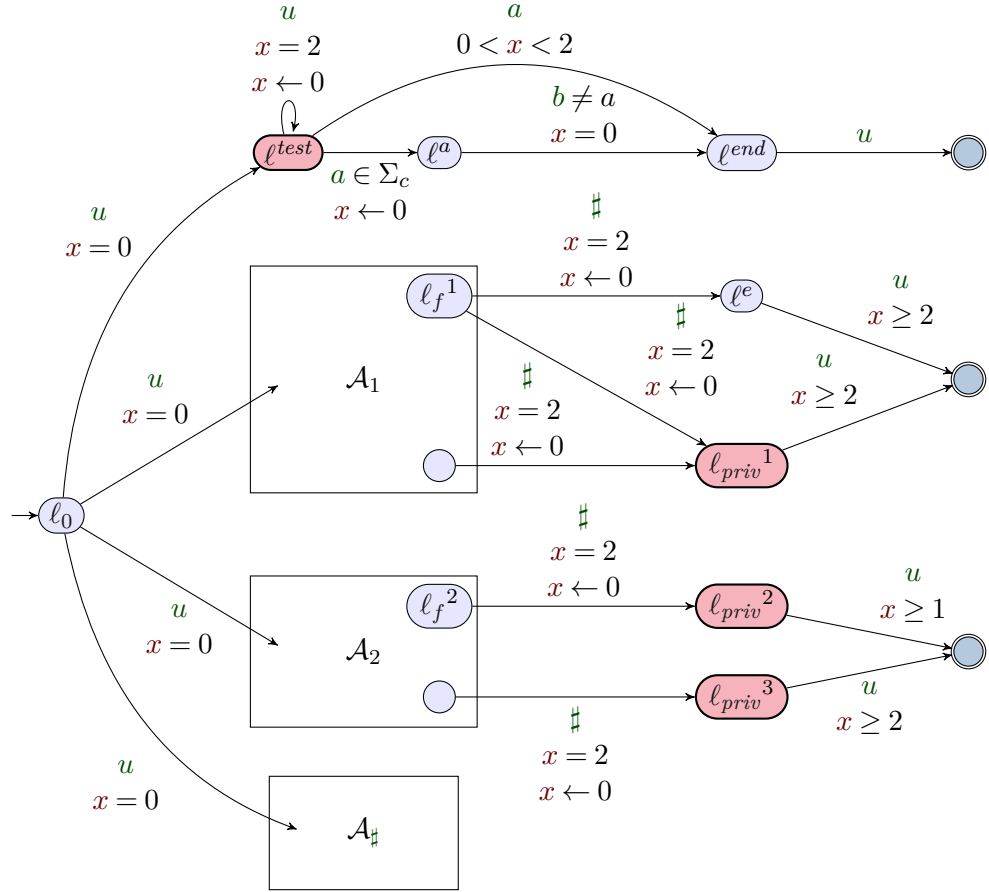

\paragraph*{First block.}

The first block mainly consists of a private location 
$\loc^{\textit{test}}$ from which one can reach another location $\loc^{\textit{end}}$ in two ways: if 
the strategy allows two actions simultaneously (type 2) 
or it enables some action at a non-even-integer time (type 1). In both cases, this will then produce private runs.
Formally, it is defined as the TA 
$\TB_0= (\ControlSet, \LocSet^0,\loc^{\textit{test}},\{\loc^{\textit{test}}\},\{\locfinal^0\},\ClockSet^0,\invariant^0,\EdgeSet^0)$ where
\begin{itemize}
\item $\LocSet^0 = \{\loc^{\textit{test}}, \loc^{\textit{end}},\locfinal^0 \} \cup \{\loc^a\mid a\in \ControlSet\}$,
\item $\ClockSet^0=\{x\}$,
\item $\invariant^0$ assigns $\top$ to every location in $\LocSet^0$,
\item $\EdgeSet^0 = \{(\loc^{\textit{test}}, x=2,\uncontrol, \{x\},\loc^{\textit{test}}), 
(\loc^{\textit{end}}, \top,\uncontrol, \emptyset,\locfinal^0)\}$\\
\makebox[15cm]{$\cup~\{(\loc^{\textit{test}}, \top, a, \{x\},\loc^{a}),(\loc^{\textit{test}}, 0<x<2,a, \emptyset,\loc^{\textit{end}}), (\loc^{a}, x=0, b, \emptyset,\loc^{\textit{end}})\mid a,b\in \ControlSet, b\neq a\}.$}
\end{itemize}

The private location $\loc^{\textit{test}}$ resets $x$ every 2 time units via a self-loop.
Whenever the strategy allows an action $a \in \ControlSet$ at a non-even-integer time $t$, the run can take the $0<x<2$ transition to $\loc^{\textit{end}}$ and subsequently reach $\locfinal^0$ at any duration, producing private runs at all times from $t$ onwards. 
Similarly, if two actions $a \neq b$ are simultaneously allowed, a run can go $\loc^{\textit{test}} \to \loc^a \to \loc^{\textit{end}}$ in zero time (as $x=0$ after the reset on the first transition), again reaching $\loc^{\textit{end}}$ 
and producing a private duration.

Conversely, if neither condition occurs (i.e., every controllable action is enabled
at most one at a time, and only at even integers) then no run ever reaches $\loc^{\textit{end}}$,
and $\TB_0$ produces no run reaching a final location (and thus produces no duration).

\smallskip\noindent
\textbf{Consequence.} 
Let $\strategy$ be a control strategy. If it allows either an action at a non-even-integer time $t$ (type 1) or two actions at the same time $t$ (type 2), then the set of private durations $\Dpriv(\AllRuns{\TB_1}{\strategy})$ of $\TB_0$ under $\strategy$ is $[t; +\infty)$. If no such time $t$ exists for $\strategy$, then $\Dpriv(\AllRuns{\TB_0}{\strategy}) = \emptyset$.

In the description of the second block, we show that those private durations violate the ET-opacity of the full TA $\TA$ if the time $t$ occurs before the action $\sharp$ is allowed.

\paragraph*{Second block.}

We now define two extended TAs that attach private/public routing to the output of
$\TA_1$ and $\TA_2$. Let
\begin{align*}
\TB^1 &= (\ControlSet,\LocSet^1,\locinit^1,\{\locpriv^1\},\{\locfinal^{\TB_1}\},
           \ClockSet_1,\invariant^1,\EdgeSet^1),\\
\TB^2 &= (\ControlSet,\LocSet^2,\locinit^2,\{\locpriv^2,\locpriv^3\},\{\locfinal^{\TB_2}\},
           \ClockSet_2,\invariant^2,\EdgeSet^2),
\end{align*}
where
\begin{itemize}
\item $\LocSet^1 = \LocSet_1 \cup \{\loc^e, \locpriv^1, \locfinal^{\TB_1}\}$ and
  $\LocSet^2 = \LocSet_2 \cup \{\locpriv^2, \locpriv^3, \locfinal^{\TB_2}\}$,
\item $\invariant^1$ and $\invariant^2$ assign $\top$ to every location,
\item $\EdgeSet^1 = \EdgeSet_1
		\cup \{(\loc, x=2, \sharp, \{x\}, \loc^e),\,
            (\loc, x=2, \sharp, \{x\}, \locpriv^1)
           \mid \loc \in \FinalSet^1\}$\\
  \makebox[8.1cm]{$\cup \{(\loc, x{=}2, \sharp, \{x\}, \locpriv^1)
           \mid \loc \notin \FinalSet^1\}$}\\
  \makebox[10.3cm]{$\cup \{(\locpriv^1, x\geq 2, \uncontrol, \emptyset, \locfinal^{\TB_1}),\,
               (\loc^e, x \geq 2, \uncontrol, \emptyset, \locfinal^{\TB_1})\}$,}\\
  and\\
  $\EdgeSet^2 = \EdgeSet_2
    \cup \{(\loc, x=2, \sharp, \{x\}, \locpriv^2)
           \mid \loc \in \FinalSet^2\}$\\
  \makebox[8.3cm]{$\cup \{(\loc, x=2, \sharp, \{x\}, \locpriv^3)
           \mid \loc \notin \FinalSet^2\}$}\\
  \makebox[10.7cm]{$\cup \{(\locpriv^2, x\geq 1, \uncontrol, \emptyset, \locfinal^{\TB_2}),\,
               (\locpriv^3, x\geq 2, \uncontrol, \emptyset, \locfinal^{\TB_2})\}$.}

\end{itemize}

Inside $\TA_1$ and $\TA_2$, the clock $x$ can be reset every time it reaches (exactly) the value $2$, but only if a transition labelled by an action in $\ActionSet$ has been taken exactly at this time. Thus, at least one such action needs to be enabled every two time units until the action $\sharp$ is enabled, allowing to exit $\TA_1$ or $\TA_2$. Otherwise, the clock $x$ will exceed $2$, being never reset. Since the transitions labelled by $\sharp$ may only be taken when $x=2$, the run will stay stuck into $\TA_1$ or $\TA_2$.

The remaining transitions of this block are uncontrollable. Hence, there are runs reaching a final location in the block $\TB_1$ or $\TB_2$ only if every two time units an action in $\ActionSet$ is enabled, until $\sharp$ is enabled at an even-integer time.

If the strategy $\strategy$ respects this condition and is not of type 1 or 2 until the action $\sharp$ has been enabled, then it is in some $S_{w\sharp}$ for some word $w \in \ActionSet^*$.

We now focus on the sets of durations produced by the runs enabled by such a strategy.




Let $\strategy$ be a strategy in $S_{w \sharp}$.
  Let $d = 2|w|$ be the time at which $\sharp$ is emitted, we have: 
    \begin{itemize}
      \item If $w \in \regexp_1$: from $\TB_1$, all times $t \geq d + 2$ are both
        private and public (opaque).
      \item If $w \notin \regexp_1$: from $\TB_1$, all times $t \geq d + 2$ are
        private only.
      \item If $w \in \regexp_2$: from $\TB_2$, all times $t \geq d + 1$ are private.
      \item If $w \notin \regexp_2$: from $\TB_2$, all times $t \geq d + 2$ are
        private only.
    \end{itemize}
    Hence, in this case, the second block does not violate ET-opacity if and only if
    $w \in \regexp_1 \setminus \regexp_2$.

\noindent\textbf{Consequence.} Strategies in some $S_{w\sharp}$ ensure weak ET-opacity of the first parts of the TA $\TA$: $\TB_0, \TB_1$ and $\TB_2$. However, at this point, a strategy which does not let any run reach a final location (\ie{} of type 3) still ensures opacity, since the sets of private and public durations are both empty in this particular case.

The third block, which is entirely uncontrollable, produces non-opaque durations, so that the ET-opacity of $\TA$ requires the reachability of the final location of the block $\TB_2$, by a public run. This prevents (type 3)-strategies from ensuring weak ET-opacity. We now detail this last block $\TB_{\sharp}$.


\paragraph*{Third block.}

The third block $\TB_{\extraAction}$ (see \cref{fig:the-final-countdown}) is a
countdown gadget in which every transition is uncontrollable, so no strategy can
influence its behaviour. We use a large constant $N$ (whose value is fixed at the end
of the proof).
The key idea is to consider $N-1$ possible executions in parallel, and analyse at which times at least one of them is public. If there is a time at which all executions are simultaneously in the private location directly leading to the final location, then this duration is private. If moreover no run reaches a final location of another block $\TB_0, \TB_1$ or $\TB_2$, then this duration reveals the private character of the run to the attacker and violates ET-opacity. More precisely, if the non-opaque duration of $\TB_{\sharp}$ is greater than twice the length of every word in $\regexp_1$ and $\regexp_2$, then every public run of $\TB_2$ (reaching a final location) could produce this duration and make it opaque. As seen in the previous paragraph, such a public run would be enabled only if the strategy is in some $S_{w\sharp}$ with $w \in \regexp_1 \setminus \regexp_2$. 

Formally, we define the countdown TA: 
$\TB_{\sharp}= (\emptyset, \LocSet_{\sharp},q_0,\{q_p\},\{q_f\},\{x,y,z\},\invariant_{\sharp},\EdgeSet_{\sharp})$ where
\begin{itemize}
\item $\LocSet_{\sharp} = \{q_0, q_1, q_<, q_=, q_p, q_f\}$,
\item $\invariant_{\sharp}$ assigns $\top$ to every location in $\LocSet_{\sharp}$,
\item $\EdgeSet_{\sharp} = \{ (q_0, y=1, u, \{y\}, q_0), \; 
  (q_0, (0<z) \land (z<N) \land (y=0), u, \{y, z\}, q_1), \;$\\
  $(q_1, z=N, u, \{z\}, q_1), \; 
  (q_1, (x=N) \land (y < N-1), u, \{x\}, q_<), \; $\\
  $(q_<, y=N-1, u, \{y\}, q_1), \;
  (q_1, (x=N) \land (y = N-1), u, \{x\}, q_=), \; 
  (q_=, z=N, u, \{y\}, q_1), \; $\\$
  (q_=, x=0, u, \emptyset, q_p), \;
  (q_<, x=0, u, \emptyset, q_p), \;  
  (q_<, x=0, u, \emptyset, q_f), \;
  (q_p, x=0, u, \emptyset, q_f)\} 
  $
\end{itemize}

\begin{figure}[bt]
\centering
\begin{tikzpicture}[pta, node distance=2cm, thin]

\node[location, initial]   (q0)   at (0, 0.5)    {$q_0$};
\node[location]            (q1)   at (3,   0.5)    {$q_1$};
\node[location]            (q2)   at (6,   3)  {$q_<$};
\node[location]            (q4)   at (6,  -2)  {$q_=$};
\node[location, private]            (priv) at (8.5,   0.5)    {$q_p$};
\node[location, final] (qf)   at (10,  3)  {$ $};
 
\draw[->] (q0) edge
    node[above, align=center] {$0 < \styleclock{z} < N\wedge \styleclock{y}=0$\\ $\styleclock{y} \assign 0$\\ $\styleclock{z} \assign 0$}
    (q1);

\draw[->] (q0) edge[loop below]
    node[below, align=center] {$\styleclock{y} = 1$\\ $\styleclock{y} \assign 0$}
    (q0);

\draw[->] (q1) edge[loop below]
    node[below=+4pt, align=center] {$\styleclock{z} = N$\\ $\styleclock{z} \assign 0$}
    (q1);
 
\draw[->] (q1) edge[bend left=20]
    node[above=+8pt, align=center] {$\styleclock{x} = N$\\ $\wedge\; \styleclock{y} < N-1$\\ $\styleclock{x} \assign 0$}
    (q2);
\draw[->] (q2) edge[bend left=20]
    node[below right=-4pt,pos =0.3, align=center] {$\styleclock{y}= N-1$\\ $\styleclock{y} \assign 0$}
    (q1);
 
\draw[->] (q1) edge[bend right=20]
    node[below left, align=center, pos = 0.8] {$\styleclock{x} = N$\\ $\wedge\; \styleclock{y} = N-1$\\ $\styleclock{x} \assign 0$}
    (q4);
\draw[->] (q4) edge[bend right=20]
    node[above right, pos=0.3, align=center] {$\styleclock{z}=N$\\ $\styleclock{y} \assign 0$}
    (q1);

\draw[->] (q2) edge
    node[above] {$\styleclock{x} = 0$}
    (qf);
 
\draw[->] (q2) edge
    node[above right] {$\styleclock{x} = 0$}
    (priv);

\draw[->] (q4) edge
    node[right] {$\styleclock{x} = 0$}
    (priv);
 
\draw[->] (priv) edge
    node[below right] {$\styleclock{x} = 0$}
    (qf);
 
\end{tikzpicture}
	\caption{The countdown gadget $\TA_{\extraAction}$}
  \label{fig:the-final-countdown}
\end{figure}

\noindent\textbf{Intuition.} 
This gadget has three clocks, $x$, $y$ and $z$. A first idea is to deterministically alternate between the locations $q_=$ and $q_<$ (through $q_1$), from which a final location is reachable in zero time. The run goes to $q_=$ if the values, modulo $N$, of $x$ and $y$ differ of exactly 1 time unit, otherwise it goes to $q_<$. 
This alternation allows to produce private and public durations if the current location is $q_<$, and only private durations if it is $q_=$. If $q_1$ were the initial location, the non-opaque durations would be the times at which the run can only be in $q_{=}, q_p$ or in the final location. However, the initial location is $q_0$, going to $q_1$ through the only non-deterministic transition, which leads us to the second idea. This transition can be taken at exactly $N-1$ times. We thus consider $N-1$ runs arriving in $q_1$ at shifted times, so that a duration is non opaque only if all runs can simultaneously be in $q_=$. The time that takes the cycle between two passages into $q_=$ for each run is constant and depends entirely on its initial shift. The clock $z$ is used to store this initial shift between $x$ and $y$ in order to recover the same configuration in $q_1$ once the location $q_=$ has been reached. 

\noindent\textbf{Arguments.}  The clock $x$ cycles from $0$ to $N$. On entry
to $q_1$, a value $x_0 \in (0,N)$ is chosen non-deterministically and stored in both
$y$ and $z$, by the constant difference $x - z = x_0$ (modulo $N$). Each cycle of duration $N$ around $q_1$, visiting either $q_<$ or $q_=$, increments $y$ by $1$ (modulo
$x_0$): the run passes through $q_<$ when $y < N-1$ (resetting $x$ at $N$) then comes back to $q_1$ after having increased $y$ by 1 by resetting it at $N-1$ instead of $N$, 
and through $q_=$ when $y = N-1$ (resetting $y$ to the value stored in $z$). 

Formally, if
$q_1$ is entered with $x = x_0$, $y = z = 0$, then $q_=$ is visited
exactly at times $dN$ where $d$ is a multiple of $x_0$.
Indeed, exactly $N$ time units elapse between two consecutive transitions from $q_1$ to $q_<$. Since at the beginning $y = x - x_0$ and since each of these passage into $q_<$ decreases of 1 the value of $x-y$ (modulo $N$), this loop requiring $y < x-1$ (modulo $N$) is taken exactly $x_0 - 1$ times. Then the condition $y = x-1$ (modulo $N$) is fulfilled and the location $q_=$ is visited the next time $x$ reaches $N$. The total time elapsed before the first time the run reaches $q_=$ is thus $x_0 N$. We recover the same configuration each time the transition from $q_=$ to $q_1$ is taken.

 Each visit to $q_=$ (resp.\ $q_<$) spawns a
private (resp.\ both private and public) run that reaches $\locfinal$ immediately.
Since $x_0$ is chosen non-deterministically over $\{1,\ldots,N{-}1\}$, the set of
durations that are private but not public from this block is 
$\{dN \mid d \text{ is a multiple of } \operatorname{lcm}(1,\ldots,N{-}1)\}$. 
Denoting $l :=\operatorname{lcm}(1,\ldots,N{-}1)$, we thus have that
the first duration at which this block violates ET-opacity is $lN$.
By the second Chebyshev function and the prime number theorem, 
$l$ is asymptotically exponential in $N$.

If a strategy waits too long before enabling the symbol $\sharp$, the block $\TB_{\sharp}$ will produce such non-opaque durations.
 Indeed, for ET-opacity to hold at duration $lN$,
the second block must produce a public run of the same duration, which requires the
strategy to have emitted $w\sharp$ with $2|w\sharp| \leq lN$ and $w \in \regexp_1$.
Combined with the second-block analysis, ET-opacity holds for the full TA if and only if the strategy is in some set $S_{w\sharp}$ with $w \in \regexp_1 \setminus \regexp_2$. Hence, there exists such an opacity-ensuring strategy if and only 
if there exists $w$ with $2|w\sharp| \leq lN$ and $w \in \regexp_1 \setminus \regexp_2$.

\paragraph*{Size bound.}
There remains to verify that if $\regexp_1 \setminus \regexp_2 \neq \emptyset$, then it
contains a word of length at most $\frac{lN}{2}$.
A regular expression with \op{} can be unrolled into an equivalent expression without
squaring of at most exponential size. Classical automata theory then gives that the
shortest word separating two regular languages is at most exponential in the size of
their automata. Hence the shortest word in $\regexp_1 \setminus \regexp_2$ is at most
doubly exponential in the size of $\regexp_1, \regexp_2$. Since $l$ is
exponential in $N$ and $N$ is written in binary (so its value is exponential in the
size of the TA), $lN$ is doubly exponential. Choosing $N$ large enough (yet of
polynomial size in the input) ensures that the shortest separating word is smaller
than $\frac{lN}{2}$, completing the reduction.	
\end{proof}

\section{Conclusion}

In this paper, we studied the complexity of execution-time opacity for timed automata, both in the uncontrolled and controlled settings, and established tight complexity bounds for both problems.
For uncontrolled timed automata, we proved that weak and full ET-opacity are \coNEXPTIME-complete (\cref{ET-opac-unc}), matching the (corrected) existing upper bound from prior work. The hardness reduction leveraged a known \NEXPTIME-hard fragment of parametric TCTL model-checking, and showed how the negation of the relevant synchronisation property can be faithfully encoded as an ET-opacity instance.
For controllable timed automata, we proved that the existence of a sequential strategy ensuring weak or full ET-opacity is \EXPSPACE-complete (\cref{th:ET-opac-comp}). The hardness reduction proceeded from the \EXPSPACE-complete inequivalence problem for regular expressions with squaring, relying on the ability to construct a doubly exponential countdown gadget.
Additionally, we showed that weak and full ET-opacity remain inter-reducible in the controlled setting (\cref{prop:weak-full}), extending a result previously established for uncontrolled timed automata. The key difficulty was ensuring that a \emph{single} sequential strategy simultaneously witnesses both inclusion directions, which required a more subtle construction than in the uncontrolled case.

These results conclude the complexity-theoretic analysis of ET-opacity for timed automata, but several research directions remain open. On the theoretical side, ET-opacity has been studied for richer models such as parametric timed automata, and existing works~\cite{ALM23,AAL24} leave interesting complexity gaps. In particular, several problems solved for weak ET-opacity remain open for full ET-opacity; the inter-reduction established in \cref{prop:weak-full} may help in closing these gaps.
On the practical side, high worst-case complexity does not necessarily preclude efficient implementations, as the instances arising in practice are often far from the worst case. A central source of complexity when analysing timed automata is the region automaton, whose size is exponential in the number of clocks and the magnitude of the constants. The zone automaton was introduced precisely to mitigate this cost, by grouping large sets of region-automaton states into symbolic zones. For the analysis of ET-opacity in uncontrolled timed automata, the zone automaton can in principle be used with mild adaptations to the existing theoretical algorithm. However, the controlled setting poses a deeper challenge: the current decision procedure relies on a determinisation of the region automaton, a construction that does not carry over naturally to the zone-based representation. Such an alternative for one-clock TA state estimation is investigated in \cite{GSLG24}. Devising a zone-compatible alternative is therefore likely to be a key step towards a practically viable implementation of the ET-opacity control problem.


	\newcommand{\CCIS}{Communications in Computer and Information Science}
	\newcommand{\CSUR}{{ACM} Computing Surveys}
	\newcommand{\ENTCS}{Electronic Notes in Theoretical Computer Science}
	\newcommand{\FAC}{Formal Aspects of Computing}
	\newcommand{\FundInf}{Fundamenta Informaticae}
	\newcommand{\FMSD}{Formal Methods in System Design}
	\newcommand{\IJFCS}{International Journal of Foundations of Computer Science}
	\newcommand{\IJSSE}{International Journal of Secure Software Engineering}
	\newcommand{\IJC}{International Journal of Control}
	\newcommand{\IseCure}{International Journal of Information Security}
	\newcommand{\IPL}{Information Processing Letters}
	\newcommand{\IC}{Information and Computation}
	\newcommand{\JAIR}{Journal of Artificial Intelligence Research}
	\newcommand{\JLAP}{Journal of Logic and Algebraic Programming}
	\newcommand{\JLAMP}{Journal of Logical and Algebraic Methods in Programming} 
	\newcommand{\JISA}{Journal of Information Security and Applications}
	\newcommand{\JLC}{Journal of Logic and Computation}
	\newcommand{\JALC}{Journal of Automata, Languages and Combinatorics}
	\newcommand{\LMCS}{Logical Methods in Computer Science}
	\newcommand{\LNCS}{Lecture Notes in Computer Science}
	\newcommand{\RESS}{Reliability Engineering \& System Safety}
	\newcommand{\RTS}{Real-Time Systems}
	\newcommand{\SCP}{Science of Computer Programming}
	\newcommand{\SOSYM}{Software and Systems Modeling ({SoSyM})}
	\newcommand{\STTT}{International Journal on Software Tools for Technology Transfer}
	\newcommand{\TCS}{Theoretical Computer Science}
	\newcommand{\TOPLAS}{{ACM} Transactions on Programming Languages and Systems ({ToPLAS})}
	\newcommand{\ToPNoC}{Transactions on {P}etri Nets and Other Models of Concurrency}
	\newcommand{\TOSEM}{{ACM} Transactions on Software Engineering and Methodology ({ToSEM})}
	\newcommand{\TSE}{{IEEE} Transactions on Software Engineering}
	\newcommand{\TCAD}{{IEEE} Transactions on Computer-Aided Design of Integrated Circuits and Systems}

	\bibliographystyle{eptcs} 
	\bibliography{PTA}

\end{document}